\documentclass[a4paper]{article}
\usepackage{xcolor}
\usepackage{authblk}
\usepackage{amsfonts}
\usepackage{amsthm,amssymb,latexsym,graphicx,amsmath,lineno}
\usepackage{algorithm}
\usepackage{algpseudocode}
\algrenewcommand\algorithmicrequire{\textbf{Input:}}
\algrenewcommand\algorithmicensure{\textbf{Output:}}

\newtheorem{theorem}{Theorem}

\newtheorem{lemma}[theorem]{Lemma}

\begin{document}
\title{Quantum algorithm for doubling the amplitude of the search problem's solution states.}
%\title{On the effect of appying Hadamard then S and then Hadamard gates on quantum computational state $|0\rangle ^{\otimes^n}$}
\author{Mauro Mezzini}
\affil{Roma Tre University, Italy}
\author{Fernando L. Pelayo}
\author{Fernando Cuartero}
\affil{University of Castilla - La Mancha, Spain}

\maketitle
\begin{abstract}
 In this paper we present a quantum algorithm which increases the amplitude of the states corresponding to the solutions of the search problem by a factor of almost two.
\end{abstract}

\section {Introduction and preliminaries}

We denote the set of natural number excluding the 0 element by $\mathbb N^+$.
If $x \in \mathbb N$, $0 \leq x < 2^n$ then we say that $|x \rangle$ is  a \emph{computational state}. We denote by $|x_{n-1}x_{n-2} \dots x_0 \rangle$ the binary representation of $|x \rangle$ where $x_{n-1}$ is the most significant bit of the binary representation of $x$.  Let also denote by $w(x) = \sum_{j=0}^{n-1} x_j$. If $z \in \mathbb N$, $0 \leq z < 2^n$ we denote by $z \cdot x$ the sum $z \cdot x= \sum_{j=0}^{n-1} x_j z_j$. Furthermore, for any $k \in \mathbb N$, $0<k<n$, we denote by $z_{\text{-k}}$ a natural number  obtained from $z$ by considering only the least $n-k$ significant bits, that is, if the binary representation of $z$ is $z_{n-1}z_{n-2}\dots  z_1z_0$ then the binary representation of $z_{\text{-k}}$ is $z_{n-k-1}z_{n-k-2}\dots z_0$. 
The $S$ gate for a single qubit is represented by the following matrix:
\begin{equation}
S = \left | \begin{array}{cc}
1 & 0\\
0 & i\\
\end{array}
\right|
\end{equation}

$\dots$\\
$\dots$

\section{The effect of applying Hadamard, S and Hadamard gates to $|0 \rangle^{\otimes^n}$}

In this section we want to determine the effect of applying the Hadamard, S and Hadamard gates on quantum state $|0\rangle ^{\otimes^n}$, that is, we want to determine a formula for 
\begin{equation} 
|\alpha \rangle= H^{\otimes^n}S^{\otimes^n}H^{\otimes^n}|0\rangle ^{\otimes^n}
\end{equation}
and we show that in the  final superposition $|\alpha \rangle=\sum_{z=0}^{2^n-1} a_z |z\rangle$ the amplitude $a_z$ of a sigle state $|z\rangle$ depends by $w(z)$.

It is known  \cite{10.5555/1972505} that given any computational state $|x\rangle$, $0\leq x < 2^n$
\begin{equation} \label{eq:hadamard}
|\psi \rangle =H^{\otimes^n}|x\rangle= \dfrac{1}{\sqrt{2^n}} \sum_{z=0}^{2^n-1} (-1)^{x \cdot z}|z\rangle 
\end{equation}
Now we start with the following Lemma
\begin{lemma} \label{lemma:application_S_port}
Let $0\leq x <2^n$
\[
S^{\otimes^n} |x \rangle = i^{w(x)}  |x_{n-1}x_{n-2} \dots x_0  \rangle
\]
\end{lemma}
\begin{proof}
By induction on $n$ being the base case with $n=1$ straightforward.
So suppose that the statement holds for $n-1$. Then
\begin{align*}
S^{\otimes^n} |x \rangle &= S|x_{n-1}\rangle \otimes S|x_{n-2}\rangle \dots \otimes S|x_0\rangle = S^{\otimes^{n-1}} | x_{n-1}\dots x_1 \rangle \otimes S|x_0\rangle =\\
&=  i^{\sum_{j=1}^{n-1} x_j}  | x_{n-1} \dots x_1 \rangle \otimes  S|x_0\rangle =  \text{ (by induction hypothesys) }\\
&= i^{\sum_{j=1}^{n-1} x_j}  | x_{n-1} \dots x_1 \rangle  \otimes  i^{x_0} |x_0\rangle = \\
&=i^{w(x)}  |x_{n-1}x_{n-2} \dots x_0  \rangle
\end{align*}
\end{proof}
\noindent
Now by Lemma \ref{lemma:application_S_port} and equation \eqref{eq:hadamard}, we have that 
\begin{align*}
|\psi_1 \rangle= S^{\otimes^n}H^{\otimes^n}|0\rangle ^{\otimes^n}=S^{\otimes^n}\dfrac{1}{\sqrt{2^n}} \sum_{x=0}^{2^n-1}  |x\rangle =\dfrac{1}{\sqrt{2^n}} \sum_{x=0}^{2^n-1} i^{w(x)} |x\rangle
\end{align*}
and applying the Hadamard to $|\psi_1 \rangle$, by \eqref{eq:hadamard}, we have that 
\begin{align*}
|\psi_2 \rangle= H^{\otimes^n}|\psi_1\rangle =\dfrac{1}{\sqrt{2^n}} \sum_{x=0}^{2^n-1} i^{w(x)}\left[ \dfrac{1}{\sqrt{2^n}}\sum_{z=0}^{2^n-1}(-1)^{x \cdot z} |z\rangle\right ]
\end{align*}
and reordering the term of the sum we have that 
\begin{align*}
|\psi_2 \rangle= \dfrac{1}{2^n}\sum_{z=0}^{2^n-1} \sum_{x=0}^{2^n-1}(-1)^{x \cdot z} i^{w(x)} |z\rangle=\dfrac{1}{2^n}\sum_{z=0}^{2^n-1} \left( \sum_{x=0}^{2^n-1}i^{w(x)} (-1)^{x \cdot z} \right) |z\rangle
\end{align*}
\noindent
So in order to compute the amplitudes of $|\psi_2 \rangle$ we need to compute the sum 
\[
\sum_{x=0}^{2^n-1} i^{w(x)} (-1)^{x \cdot z} 
\]
for every $0 \leq z<2^n$. We will do this in  the following two theorems.
%Let $0 \leq z <2^{2(m+1)}$. If $z_{2m+1}z_{2m}\dots z_0$ is the binary representation of $z$ we define $\bar{z} \in \mathbb N$ a number obtained from $z$ by considering only the least significant $2m$ bits of $z$. That is, the binary representation of $\bar{z}$ is $z_{2m-1}z_{2m-2}\dots z_0$. 
First of all we need the following Lemma.

\begin{lemma} \label{lemma:preliminary_lemma}
Let $0 \leq z< 2^{2m+1}$ and $0 \leq x< 2^{2m+1}$ and  let $z_{2m}z_{2m-1}\dots z_0$ and  $x_{2m}x_{2m-1}\dots z_0$ be the binary representation, respectively, of $z$ and $x$. We have that
\begin{align} \label{eq:prliminary}
\sum_{x=2^{2m}}^{2^{2m+1}-1} &i^{w(x)}(-1)^{\sum_{j=0}^{2m} z_j \cdot x_j}= i(-1)^{z_{2m}} \sum_{x=0}^{2^{2m}-1} i^{w(x)}(-1)^{\sum_{j=0}^{2m-1} z_j \cdot x_j} 
\end{align} 
\end{lemma}
\begin{proof}
We note that, on the left hand of the equation \eqref{eq:prliminary}, for every element of the sum, we have that $x_{2m}=1$. Therefore  $\sum_{j=0}^{2m} z_j \cdot x_j= \sum_{j=0}^{2m-1} z_j \cdot x_j +z_{2m}$. Based on this we have that 
\begin{align*}
\sum_{x=2^{2m}}^{2^{2m+1}-1} &i^{w(x)}(-1)^{\sum_{j=0}^{2m} z_j \cdot x_j}= (-1)^{z_{2m}} \sum_{x=2^{2m}}^{2^{2m+1}-1} i^{w(x)}(-1)^{\sum_{j=0}^{2m-1} z_j \cdot x_j} 
\end{align*} 
Furthermore for the same reason above, if $2^{2m}\leq x <2^{2m+1}$ and  if $0\leq \bar{x} <2^{2m}$ then we have that $w(x)=w(\bar{x})+1$ and this prove the equation \eqref{eq:prliminary}.
\end{proof}

\begin{theorem} \label{lemma:first_lemma}
Let $0 \leq z < 2^n$, $z \in \mathbb N$. If $n=2m$ is even we have that
\begin{equation} \label{eq:first_lemma}
\sum_{x=0}^{2^n-1} i^{w(x)}(-1)^{z \cdot x} =  (-1)^{w(z)}i^{m+w(z)}2^m 
\end{equation}
\end{theorem}
\begin{proof}
We prove the equation \eqref{eq:first_lemma} on induction on $m$ being the base case with $m=1$ easily verifiable for all $z \in \{0,1,2,3\}$.
So suppose the statement holds for all $h \leq m$ and for all $0 \leq z < 2^{2m}$. Then, for any $0\leq z <2^{2m+2}$ we have  
\begin{align}
\sum_{x=0}^{2^{2m+2}-1} i^{w(x)}(-1)^{z \cdot x} =&\sum_{x=0}^{2^{2m}-1} i^{w(x)}(-1)^{\sum_{j=0}^{2m-1} z_j \cdot x_j} \label{first_term}\\
+&\sum_{x=2^{2m}}^{2^{2m+1}-1} i^{w(x)}(-1)^{\sum_{j=0}^{2m} z_j \cdot x_j}+\label{second_term}\\
+&\sum_{x=2^{2m+1}}^{2^{2m+2}-1} i^{w(x)}(-1)^{\sum_{j=0}^{2m+1} z_j \cdot x_j}\label{third_term}
\end{align}
%Now if $z_{2m+1}z_{2m}\dots z_0$ is the binary representation of $z$ we define $z_{\text{-2}} \in \mathbb N$ a number obtained from $z$ by considering only the least significant $2m$ bits of $z$. That is, the binary representation of $z_{\text{-2}}$ is $z_{2m-1}z_{2m-2}\dots z_0$. 
Now, by equation \eqref{eq:prliminary} and  by induction hypothesys, we have that \eqref{second_term} is equal to 
\begin{align}
\sum_{x=2^{2m}}^{2^{2m+1}-1} &i^{w(x)}(-1)^{\sum_{j=0}^{2m} z_j \cdot x_j}=
 i(-1)^{z_{2m}} \sum_{x=0}^{2^{2m}-1} i^{w(x)}(-1)^{\sum_{j=0}^{2m-1} z_j \cdot x_j}= \nonumber\\
&=i(-1)^{z_{2m}} (-1)^{w(z_{\text{-2}})}i^{m+w(z_{\text{-2}})}2^m  \label{final_second_term}
\end{align}
Likewise, in the term \eqref{third_term}, for each $x$ is the sum, the bit $x_{2m+1}$ is always set to 1, so we have that \eqref{third_term} is, by  equation \eqref{eq:prliminary}, equal to 
\begin{align}
\sum_{x=2^{2m+1}}^{2^{2m+2}-1}& i^{w(x)}(-1)^{\sum_{j=0}^{2m+1} z_j \cdot x_j}=i(-1)^{z_{2m+1}} \sum_{x=0}^{2^{2m+1}-1} i^{w(x)}(-1)^{\sum_{j=0}^{2m} z_j \cdot x_j} \label{final_third_term}
\end{align}
Now by repeatedly applying equation \eqref{eq:prliminary} and the induction hypothesys we have that the sum in right hand of equation \eqref{final_third_term} is
\begin{align}
\sum_{x=0}^{2^{2m+1}-1} & i^{w(x)}(-1)^{\sum_{j=0}^{2m} z_j \cdot x_j}=\\
=& \sum_{x=0}^{2^{2m}-1} i^{w(x)}(-1)^{\sum_{j=0}^{2m-1} z_j \cdot x_j}+\sum_{x=2^{2m}}^{2^{2m+1}-1} i^{w(x)}(-1)^{\sum_{j=0}^{2m} z_j \cdot x_j}=\nonumber\\
&= (-1)^{ w(z_{\text{-2}}) }i^{m+w( z_{\text{-2}} ) }2^m+i(-1)^{z_{2m}}\sum_{x=0}^{2^{2m}-1} i^{w(x)}(-1)^{\sum_{j=0}^{2m-1} z_j \cdot x_j}=\nonumber\\
&=(-1)^{w(z_{\text{-2}})}i^{m+w(z_{\text{-2}})}2^m \left[ 1 +i(-1)^{z_{2m}}\right] \label{eq:insidesum}
\end{align}

So if we replace \eqref{eq:insidesum} in \eqref{final_third_term} and if we  sum togheter \eqref{first_term}, \eqref{final_second_term} and \eqref{final_third_term} we obtain 

\begin{align}
a_z =&(-1)^{w(z_{\text{-2}})}i^{m+w(z_{\text{-2}})}2^m \left[ 1+i(-1)^{z_{2m}}+i(-1)^{z_{2m+1}}+i^2 (-1)^{z_{2m}+z_{2m+1}}\right]=\nonumber\\
=&(-1)^{w(z_{\text{-2}} )}i^{m+1+w(z_{\text{-2}})}2^m \left[-i+(-1)^{z_{2m}}+(-1)^{z_{2m+1}}+i (-1)^{z_{2m}+z_{2m+1}}\right] \label{final_formula}
\end{align}
Now if we denote by $P=(-1)^{w(z_{\text{-2}})}i^{m+1+w(z_{\text{-2}})}2^m$ we have that \eqref{final_formula} is
\begin{align*}
a_z= \left \{ \begin{array}{rl}
2 P & \text{ if $z_{2m}=z_{2m+1}=0$ }\\
-2i P & \text{ if $z_{2m} \neq z_{2m+1}$ }\\
-2 P & \text{ if $z_{2m}=z_{2m+1}=1$ }\\
\end{array} \right .
\end{align*}
and it is now easy to verify that 
\[
a_z= (-1)^{w(z)} i^{m+1+w(z)}2^{m+1}
\]
for every $0\leq z <2^{2m+2}$, and this proves the induction step.
\end{proof}

\begin{theorem} \label{lemma:last_lemma}
Let $n=2m+1$ an odd natural, $m \in \mathbb N$ and let $0 \leq z < 2^n$, $z \in \mathbb N$.
%If $z_{2m}z_{2m-1}\dots z_0$ is the binary representation of $z$ we define $z_{\text{-1}} \in \mathbb N$ a number obtained from $z$ by considering only the least significant $2m$ bits of $z$. That is, the binary representation of $z_{\text{-1}}$ is $z_{2m-1}z_{2m-2}\dots z_0$. 
Then
\begin{equation} \label{eq:odd_case}
\sum_{x=0}^{2^n-1} i^{w(x)}(-1)^{z \cdot x} =  (-1)^{w(z)}i^{m+w(z)}2^m(1+i)
\end{equation}
\end{theorem}
\begin{proof}
First of all we note that equation \eqref{eq:odd_case} holds if $m=0$ and $z \in \{0,1\}$. So in the following we suppose that $m>1$. We have that
\begin{align*}
a_z &=\sum_{x=0}^{2^{2m+1}-1} i^{w(x)}(-1)^{z \cdot x} =\nonumber\\
=&\sum_{x=0}^{2^{2m}-1} i^{w(x)}(-1)^{\sum_{j=0}^{2m-1} z_j \cdot x_j} 
+\sum_{x=2^{2m}}^{2^{2m+1}-1} i^{w(x)}(-1)^{\sum_{j=0}^{2m} z_j \cdot x_j}
\end{align*}
and, by Theorem \ref{lemma:first_lemma}, and by equation \eqref{eq:prliminary}, we have
\begin{align} \label{eq:final_intermediary}
a_z &=(-1)^{w(z_{\text{-1}})}i^{m+w(z_{\text{-1}})}2^m+i(-1)^{z_{2m}} \sum_{x=0}^{2^{2m}-1} i^{w(x)}(-1)^{\sum_{j=0}^{2m-1} z_j \cdot x_j}=\nonumber\\
=&(-1)^{w(z_{\text{-1}})}i^{m+w(z_{\text{-1}})}2^m+i(-1)^{z_{2m}}(-1)^{w(z_{\text{-1}})}i^{m+w(z_{\text{-1}})}2^m=\nonumber\\
=&  (-1)^{w(z_{\text{-1}})}i^{m+w(z_{\text{-1}})}2^m \left [ 1+i(-1)^{z_{2m}} \right]
\end{align}
Let $z_{2m}z_{2m-1} \dots z_0$ be the binary representation of $z$. Suppose first  that $z_{2m}=0$. Then equation \eqref{eq:final_intermediary} become
\begin{equation}
(-1)^{w(z)}i^{m+w(z)}2^m+ (-1)^{w(z)}i^{m+w(z)+1}2^m
\end{equation}
and the Theorem is therefore proved. So suppose that $z_{2m}=1$. Then equation \eqref{eq:final_intermediary} become
\begin{equation}
 (-1)^{w(z)-1}i^{m+w(z)-1}2^m+ (-1)^{w(z)}i^{m+w(z)}2^m
\end{equation}
but observing that 
\begin{equation}
(-1)^{w(z)}i^{m+w(z)+1}=(-1)^{w(z)-1}i^{m+w(z)-1}
\end{equation}
we have that also in this case the Theorem is satisfied.
\end{proof}

As an example we have computed the amplitudes $a_z$ (disregarding the normalization factor) for $n \in {3,4}$ and we report them on Table \ref{tbl:amplitude_example}.

\begin{table}
\begin{center}
\begin{tabular}{|c|c|}
\hline
$\mathbf{|z\rangle }$ & $\mathbf{a_z}$ \\
\hline
$|000\rangle$ & $-2+2i$\\ \hline
$|001\rangle$ & $2+2i$\\ \hline
$|010\rangle$ & $2+2i$\\ \hline
$|011\rangle$ & $2-2i$\\ \hline
$|100\rangle$ & $2+2i$\\ \hline
$|101\rangle$ & $2-2i$\\ \hline
$|110\rangle$ & $2-2i$\\ \hline
$|111\rangle$ & $-2-2i$\\ \hline
\end{tabular}
\begin{tabular}{|c|c|}
\hline
$\mathbf{|z\rangle }$ & $\mathbf{a_z}$ \\
\hline
$|0000\rangle$ & $-4$\\ \hline
$|0001\rangle$ & $4i$\\ \hline
$|0010\rangle$ & $4i$\\ \hline
$|0011\rangle$ & $4$\\ \hline
$|0100\rangle$ & $4i$\\ \hline
$|0101\rangle$ & $4$\\ \hline
$|0110\rangle$ & $4$\\ \hline
$|0111\rangle$ & $-4i$\\ \hline
$|1000\rangle$ & $4i$\\ \hline
$|1001\rangle$ & $4$\\ \hline
$|1010\rangle$ & $4$\\ \hline
$|1011\rangle$ & $-4i$\\ \hline
$|1100\rangle$ & $4$\\ \hline
$|1101\rangle$ & $-4i$\\ \hline
$|1110\rangle$ & $-4i$\\ \hline
$|1111\rangle$ & $-4$\\ \hline
\end{tabular}
\end{center}
\caption{Left the amplitudes of $a_z$ for $n=3$. Right the amplitudes of $a_z$ for $n=4$. In order to get the final amplitudes one should multiply them by a suitable normalization factor. \label{tbl:amplitude_example}}
\end{table}

\section{Doubling the amplitude of the search problem's solution states}

In this section we consider a quantum circuit for doubling the amplitude of solution's  states of the search problem.
For a search problem we refer, in general, to the problem of finding a solution of some NP-complete problem.
Like in the Groover algorithm we will use the intrinsic quantum mechanical parallelism and an oracle $f(|x \rangle)\in \{0,1\}$, specifically designed for the specific problem at hand, which return $1$ is $x$ is a solution of the problem and $0$ otherwise. 

In particular, in order to present in the detail the results of this paper, we will use a quantum oracle for the Partition Problem (PP). 
In the PP  we have a finite set of elements $E$ and a function $s: E \rightarrow \mathbb N^+$. We want to find a subset $E' \subset E$ such that $\sum_{e \in E'} s(e) = \sum_{e \in E \setminus E'} s(e)$. From now on we do not loss generality if we consider the set $E$ equal to the set of the first $|E|$ naturals, that is we always consider $E = \{0,1,\dots , n-1\}$. Furthermore we note that if PP has a solution $E'$ then $E - E'$ is also a solution of the PP:

The partition problem (PP) is well known to be an NP-complete problem \cite{Kar72}.\\

%In this section we consider an application of the gates described in the previous section to the solution of the PP. 

We describe, in the following, an application of the gates described in the previous section in a quantum circuit to deal with PP (see Figure \ref{fig:circuit}). 
While the following results apply specifically to the PP they can be applied to any other search problem. 

Denote by   $\mathcal S = \sum_{e \in E} s(e)/2$. Recall that PP problem has a solution only if $\mathcal S$ is an integer.
We use the two's complements representation of $-\mathcal S $ requiring $m=\lceil \log_2{\mathcal S}\rceil +1$ qubits. Then for each $e \in E$, we use $k_e=\lceil \log_2{s(e)}\rceil+1$  qubits to encode $s(e)$. These qubits will remain constant in every phase of the circuit and therefore we will not consider them in the reasoning that follows.  
We use $n$ qubits to encode a subset $E'$ of $E$. If $|x_{n-1}x_{n-2}\dots x_0\rangle$ is the state of those $n$ qubits, then an element $e$, $0 \leq e < n$, is  included in the set $E'$ if and only if $x_e=1$. 
 We will use $m$ qubits, denoted in the following by $|\sigma\rangle$, to store the sum $\sigma = - \mathcal S +\sum_{e \in E'}s(e)$ for the elements selected in $|x\rangle$. In this way $|\sigma\rangle= |0\rangle^{\otimes^m}$ for a solution $|x\rangle$ of the PP. 
We also use a control qubit $|c\rangle$.

So we have four groups of bits: $|x\rangle$, $|\sigma\rangle$, $|c\rangle$ and the sets of qubits used to represent the constants $s(e)$ for each element of $E$. Note that the number of qubits of the circuit, $n+m+1+\sum_{e \in E} k_e$, is polynomial in the size of a coincise specification of the PP. 

At the beginning of the circuit we have the following superposition:
\begin{align*}
|\varphi_0 \rangle = |0 \rangle^{\otimes^n} |\sigma\rangle|c\rangle
\end{align*}
where $\sigma$ is set to the two's complement of $-\mathcal S $ and $|c\rangle$ is set to $|1\rangle$. Then, we apply the Hadamard gate to the first $n$ qubits obtaining
\begin{align*}
|\varphi_1 \rangle &= \left (H^{\otimes^n}\otimes I^{m+1} \right)|\varphi_0 \rangle = \dfrac{1}{\sqrt{2^n}} \sum_{x=0}^{2^n-1} |x\rangle |\sigma\rangle|c\rangle
\end{align*}
Next, we uses each qubit $x_e$ to conditionally sum the element $s(e)$ to $|\sigma \rangle$. If there exist a solution to the PP then, in the final superposition of $|\sigma \rangle$, the amplitude of the state $|x\rangle|0 \rangle^{\otimes^m} |c\rangle$ will be not $0$. The states $|x\rangle$ for which $|\sigma \rangle$ is zero wil be referred as the \emph{solutions states} of the PP. The control qubit $|c\rangle$ will be set to zero exactly for those states for which $|\sigma \rangle=|0\rangle^{\otimes^m}$. At this point we apply an uncomputational step in order to set  $|\sigma\rangle=|-\mathcal S\rangle $. Now if we apply the $S$ gate to the first $n$ qubits we obtain, by Lemma \ref{lemma:application_S_port}
\begin{align*}
|\varphi_2 \rangle &=\left ( S^{\otimes^n}\otimes I^{m+1} \right)|\varphi_1 \rangle =  \dfrac{1}{\sqrt{2^n}} \sum_{x=0}^{2^n-1} i^{w(x)}|x\rangle |\sigma\rangle|c\rangle
\end{align*}
After this we apply again the Hadamard gate to the first $n$ qubits. This operation is controlled by the control qubit in a way that the Hadamard port is applied only to non solution states. For the sake of simplicity we suppose in the following, that the PP has only two solution whose numeric representation are $y$ and its bitwise complement $\overline y$.  By equation \eqref{eq:prliminary}, we obtain
\begin{align}
|\varphi_2 \rangle &\xrightarrow{\text{contr } H^{\otimes^n} }|\varphi_3 \rangle =  \dfrac{1}{\sqrt{2^n}} \sum_{z \in \{y, \overline y\}} i^{w(z)} |z\rangle |\sigma\rangle|0\rangle+  \nonumber \\
+& \dfrac{1}{2^n}\sum_{z=0}^{2^n-1}  \sum_{x \notin \{y, \overline y\}} i^{w(x)}(-1)^{x\cdot z}|z\rangle  |\sigma\rangle|1\rangle=\nonumber \\
=&\dfrac{1}{2^n}\left[ \sqrt{2^n}\sum_{z \in \{y, \overline y\}} i^{w(z)} |z\rangle|\sigma\rangle|0\rangle + 
\sum_{z=0}^{2^n-1}  \sum_{x \notin \{y, \overline y\}} i^{w(x)}(-1)^{x\cdot z}|z\rangle|\sigma\rangle|1\rangle \right]  \label{eq:final_amplitude_solution}
\end{align}
Now we want to quantify the amplitude of the state $|y\rangle|\sigma\rangle|1\rangle$ and $|\overline y\rangle|\sigma\rangle|1\rangle$ of equation \eqref{eq:final_amplitude_solution}. We consider only the  state $|y\rangle|\sigma\rangle|1\rangle$ since the same arguments can be applied to state $|\overline y\rangle|\sigma\rangle|1\rangle$.  The amplitude $b_y$ fr the state $|y\rangle|\sigma\rangle|1\rangle$ (in the following we disregard the normalization factor $1/2^n$) is given by the following formula
\begin{align}
b_y =\sum_{x \notin \{y, \overline y\}} i^{w(x)}(-1)^{x\cdot y}  \label{eq:final_amplitude_solution_for_a_signle_solution}
\end{align}
We may write the above sum as
\begin{align}
 b_y= \sum_{x \notin \{y, \overline y\}} i^{w(x)}(-1)^{x\cdot y}=\sum_{x=0}^{2^n} i^{w(x)}(-1)^{x\cdot y}- 
\sum_{x\in{\{y, \overline{y}\}}} i^{w(x)}(-1)^{x\cdot y}
\end{align}
We have that 
\begin{align}
\sum_{x\in{\{y, \overline{y}\}}} i^{w(x)}(-1)^{x\cdot y}&=i^{w(y)}(-1)^{y\cdot y} + i^{n-w(y)}(-1)^{\overline{y} \cdot y}=\nonumber\\
=&i^{w(y)}(-1)^{{w(y)}} + i^{n-w(y)}=\label{eq:generic_z_term_first}
%\\=&i^{w(y)}[(-1)^{{w(y)}} + i^{n-2w(y)}] \label{eq:generic_z_term}
\end{align}
Then, recalling that $i^x=i^{-x}$ when $x$ is even and $i^{-x}=-i^x$ when $x$ is odd, we have two cases: $w(y)$ is even and then 
\begin{align}
i^{w(y)}(-1)^{{w(y)}} &+ i^{n-w(y)}=i^{w(y)}(-1)^{{w(y)}} + i^{n+w(y)}=\nonumber\\
&=i^{w(y)}(1+ i^n)
\end{align}
while if $w(y)$ is odd 
\begin{align}
i^{w(y)}(-1)^{{w(y)}} &+ i^{n-w(y)}=i^{w(y)}(-1)^{{w(y)}} - i^{n+w(y)}=\nonumber\\
&=-i^{w(y)}(1+ i^n)
\end{align}

For simplicity of notation in the following we denote $w(y)$ as simply $\overline{w}$. We have that if $n=2m$ is even then, by Theorem \ref{lemma:first_lemma}, $b_y$ is 
\begin{align} \label{eq:duplicate_n even}
b_y= \left\{   \begin{array}{ll}
(-1)^{\overline{w}} i^{m+{\overline{w}}} 2^m-  i^{\overline{w}}(1+i^{2m})& \text{ if $\overline{w}$ is even}\\
& \\
 (-1)^{\overline{w}} i^{m+{\overline{w}}} 2^m +  i^{\overline{w}}(1+i^{2m}) & \text{ if $\overline{w}$ is odd}\\
\end{array} \right .
\end{align}
while if $n=2m+1$ is odd, by Theorem \ref{lemma:last_lemma},  $b_y$ is
\begin{align}\label{eq:duplicate_n odd}
b_y= \left\{   \begin{array}{ll}
(-1)^{\overline{w}}i^{m+\overline{w}}2^m(1+i)- i^{\overline{w}}(1+i^{2m+1})& \text{ if $\overline{w}$ is even}\\
& \\
(-1)^{\overline{w}}i^{m+\overline{w}}2^m(1+i)+ i^{\overline{w}}(1+i^{2m+1}) & \text{ if $\overline{w}$ is odd}\\
\end{array} \right .
\end{align}
It is immediate to check that in the above equations \eqref{eq:duplicate_n even} and \eqref{eq:duplicate_n odd} the term  $i^{\overline{w}}(1+i^{n})$ become trascurable, with respect to the other term in the equation, as $m$ become bigger. We conclude that the amplitude of the state $|y\rangle|\sigma\rangle|1\rangle$ is almost the same of the amplitude of state $|y\rangle|\sigma\rangle|0\rangle$, thus effectively duplicating the chances of state $|y\rangle$ at the end of the circuit.
For example if $n=2m+1=3$ and $|y\rangle =|011\rangle$ we have that $b_y=3-3i$, so that the probability of getting $|y\rangle$ is, by \eqref{eq:final_amplitude_solution}, $\dfrac{1}{64}\left [|2\sqrt{2}|^2+|3-3i|^2 \right ]=\dfrac{26}{64}=.40625$ which is exactly the output of the quirk simulator.
\begin{figure}[htbp]
    \centering
    \includegraphics[width=7.6cm]{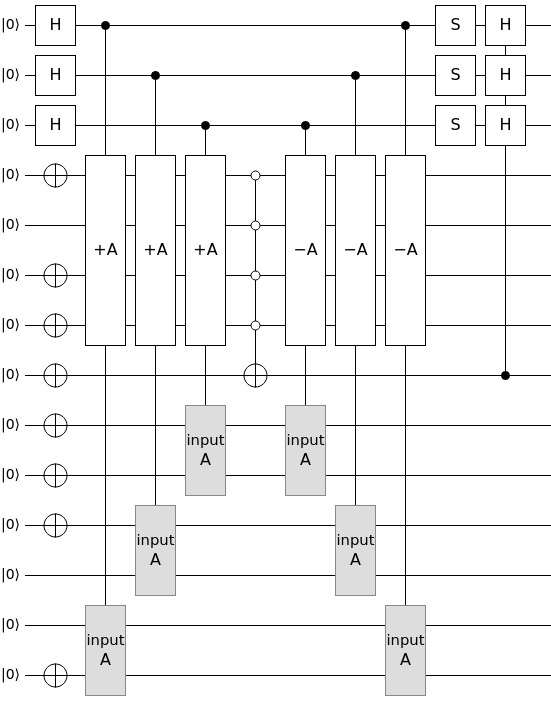}
  \caption{The circuit exploliting $S$ gates.}
  \label{fig:circuit}
\end{figure}

\section{Conclusion and future work}

We presented here a quantum algorithm for doubling the amplitude of the state correspondig to the solution of the partition problem. This is interesting because if we would be able to iterate such a doubling we could provide a polynomial quantum algorithm for solving an NP-complete problem. 
Possible future works should focus on: generalizing the mathematical results to istance of search problems where there are more than 2 solutions, find out if this algorithm can be combined to the Goover algorithm in order to seed up the latter of a factor of $p$ where $p\geq 2$ and, more important, to check if and how  it is possibile to iterate the doubling of the amplitude in order to get some polynomial time algorithm for solving the search problem.

\clearpage
\bibliographystyle{abbrv}
\bibliography{notes_on_HSHx}

\end{document}